\documentclass{article}
\usepackage[a4paper]{geometry}
\usepackage[utf8]{inputenc}                  
\usepackage[english]{babel}		   
\usepackage[font=scriptsize]{caption}
\usepackage{epsfig,graphicx}
\usepackage{amsmath,amssymb, amsthm} 
\usepackage{subfig}									
\usepackage{hyperref}								
\usepackage{braket}                          
\usepackage{tikz,pgf}			               
\usepackage{epigraph}
\usepackage{lmodern}
\usepackage{booktabs}
\usepackage{authblk}

\numberwithin{equation}{section}

\def\Q{\mathcal{Q}}
\def\N{\mathbb{N}}
\def\R{\mathbb{R}}

\def\Z{\mathbb{Z}}
\def\S{\mathbb{S}}
\def\A{{\cal A}}
\def\Q{{\cal Q}}

\DeclareMathOperator{\Id}{I}

\newcommand{\norm}[1]{\left\lVert#1\right\rVert}      

\theoremstyle{theorem}
\newtheorem{theorem}{Theorem}[section]
\newtheorem{lemma}[theorem]{Lemma}

\theoremstyle{definition}
\newtheorem{definition}[theorem]{Definition}
\newtheorem{remark}[theorem]{Remark}

\title{\bf Local minimality properties of circular motions in $1/r^\alpha$ potentials and
of the figure-eight solution of the 3-body problem}

\author[1,2]{M. Fenucci\thanks{mail: \texttt{fenucci@mail.dm.unipi.it}}}
\affil[1]{Department of Astronomy, Faculty of Mathematics, University of Belgrade, Studentski trg 16, 11000 Belgrade, Serbia}
\affil[2]{Dipartimento di Matematica, Università di Pisa, Largo B. Pontecorvo 5, 56127 Pisa, Italy}

\begin{document}
\maketitle
\begin{abstract}
   We first take into account variational problems with periodic boundary
   conditions, and briefly recall some sufficient conditions for a periodic solution
   of the Euler-Lagrange equation to be either a directional, a weak, or a strong local minimizer. 
   We then apply the theory to circular orbits of the Kepler problem with potentials
   of type $1/r^\alpha, \, \alpha > 0$. By using numerical computations, we show that
   circular solutions are strong local minimizers for $\alpha > 1$, while they are saddle
   points for $\alpha \in (0,1)$. Moreover, we show that for $\alpha \in (1,2)$ the
   global minimizer of the action over periodic curves with degree $2$ with respect to the
   origin could be achieved on non-collision and non-circular solutions.  
   After, we take into account the figure-eight solution of the 3-body
   problem, and we show that it is a strong local minimizer over a particular 
   set of symmetric periodic loops.

\end{abstract}
\section{Introduction}
In recent years, new periodic solutions of the Newtonian $N$-body problem have been
discovered by means of variational methods. In particular, taking into account $N$ equal unitary masses
and denoting by $u=(u_1,\dots,u_N):[0,T] \to \R^{3N}$ their motion, periodic orbits
are found as minimizers of the Lagrangian action functional
\begin{equation}
   \A(u) = \int_{0}^{T}\bigg( \frac{1}{2} \sum_{i=1}^N |\dot{u}_i|^2 + \sum_{1\leq i < j
   \leq N} \frac{1}{|u_i-u_j|} \bigg)dt,
   \label{eq:NbodyAction}
\end{equation}
on a set $X$ of $T$-periodic loops, see for instance \cite{chenciner-montgomery_2000, chenciner-venturelli_2000, chenciner_2002, 
chenciner_2005, ferrario-terracini_2004, barrabes-etal_2006, fusco-gronchi-negrini_2011}, 
and references therein. 
Numerical techniques have been developed as well for the computation of these orbits, 
see for instance \cite{moore_1993, simo_2000, simo_2001, simo_2002, 
kapela-zgli_2002, moore-nauenberg_2006, kapela-simo_2007, kapela-simo_2017, fenucci-gronchi_2018} 
and references therein. 
However, as already noticed in \cite{simo_2000}, despite 
the effectiveness of these methods in computing such orbits, they do not ensure that what is 
computed is actually a minimizer of the action, not even locally. For this reason, we
are interested in finding conditions that guarantee the local minimality, that can be used
to understand what kind of stationary point has been computed.

In Calculus of Variations, when the action functional is defined on a subset of $C^1$
curves $u:[0,T]\to\R^n$ such that 
\begin{equation} 
   \begin{cases} 
      u(0) = u_0, \\ u(T) = u_T, 
   \end{cases}
   \label{eq:fixedEndProblem} 
\end{equation} 
where $u_0,u_T \in \R^n$ are fixed points, the problem is usually called \textit{fixed end-point problem} or \textit{problem
of Bolza}, and the theory of local minimizers is a very well-known topic. Indeed, the
mathematical formulation of the fixed end-point problem started between the 17th and the 18th centuries,
when Bernoulli, Leibniz and Newton independently studied the \textit{brachistochrone
problem} \cite{phillips_1967, haws-kiser_1995}. Further developments have been done in 
the late 18th century and all over the entire 19th century by, for instance, Euler,
Lagrange, Legendre, Jacobi, and Weierstrass. In the last century, the fixed end-point
problem became a standard topic in Calculus of Variations, being the subject of several
classical mathematical textbooks (see for instance \cite{pars_1962, giaquinta_2004,
vanbrunt_2003}).  The interested reader can refer to \cite{goldstine_1980} for a detailed
chronological history of Calculus of Variations. 
%
When the boundary conditions \eqref{eq:fixedEndProblem} change, definitions and proofs
have to be adapted and, depending on the problem one is facing, different necessary and
sufficient conditions arise.
Theories for the positivity of quadratic functionals with disjoined boundary
conditions\footnote{For disjoined boundary conditions we mean that they can be expressed
   through two different equations $\phi_0(u(0)) = 0, \, \phi_T(u(T)) = 0$, while with
   general boundary conditions we mean that they are expressed as $\phi(u(0),u(T)) = 0$.}
   can be found in \cite{zeidan_1992, dosla-zezza_1994}, while the case of general
   boundary conditions is treated in \cite{zeidan-zezza_1989, dosla-zezza_1994,
   dosla-dosly_1997, stefanini-zezza_1997, zeidan_1999}. Moreover, in
   \cite{allwright-vinter_2005} the authors give second order minimality conditions for
   periodic optimal control problems.
A theory of weak local minimizers for different boundary condition types can be found in
\cite{hilscher-zeidan_2008,  hilscher-zeidan_2008c}, and references therein.
A theory of strong local minimizers for disjoined boundary conditions has been developed in
\cite{zeidan_1993}, and further improvements led to a theory for general boundary
conditions in \cite{zeidan_2001}. These two works are both based on the existence of a symmetric
solution of a Riccati differential equation which satisfies a specific boundary condition.

In this work, we first recall some sufficient conditions for the local minimality in 
variational problems with periodic boundary conditions. In particular, we adapt to the
periodic case the proof given in \cite{clarke-zeidan_1986} for the fixed end-point
problem, by providing appropriate boundary conditions of the symmetric solution of the Riccati
differential equation.
%
After, we use the theory in two problems of Celestial Mechanics.
First, we take into account the circular solutions in potentials of type $1/r^\alpha$,
$\alpha>0$, where $r$ is the distance from the center.  
By using numerical computations, we show that they are strong local minimizers 
for $\alpha>1$, and saddle points for $\alpha \in (0,1)$. 
Moreover, we present an example with $\alpha \in (1,2)$ where the
global minimizer of the action over periodic curves with degree $2$ with respect to the
origin is achieved on a non-collision and non-circular solutions. 
Then, we take into account the figure-eight solution of the 3-body problem
(see \cite{chenciner-montgomery_2000, moore_1993, simo_2002}). 
By using numerical computations, we show that the figure-eight is not optimal over 
the entire set of periodic loops, but it becomes a strong local minimizer when additional 
symmetries are taken into account.
\section{Definitions of local minimizers}
\label{s:defs}
Let $T>0$, and consider a functional 
\begin{equation}
   \mathcal{A}(u) = \int_{0}^{T}L(t,u,\dot{u}) dt,
   \label{eq:funct}
\end{equation}
where $L: [0,T] \times \Omega \to \R$ is a $C^2$ function, $T$-periodic in the variable
$t$, and $\Omega \subseteq \R^n \times \R^n$ is an open set. 
We denote the space of the $C^1$ $T$-periodic functions with
\begin{equation}
   V = \{ u \in C^1([0,T], \R^n) : u(0) = u(T) \},
   \label{eq:TperiodicFunctions}
\end{equation}
and assume that $\A$ is defined on a set $X \subseteq V$.
We say that $u_0 \in X$ is a
\begin{itemize}
   \item[(GM)] \textit{global minimum point} if $\A(u) \ge \A(u_0)$ for all $u \in X$;
   \item[(SLM)] \textit{strong local minimum point} if there exists $\varepsilon > 0$ such
       that, for all $u \in X$ satisfying
      \[
         \norm{u-u_0}_\infty < \varepsilon,
      \]
      we have that $\A(u) \ge \A(u_0)$.
   \item[(WLM)] \textit{weak local minimum point} if there exists $\varepsilon > 0$ such
      that, for all $u \in X$ satisfying
      \[
         \norm{u-u_0}_\infty + \norm{\dot{u}-\dot{u}_0}_\infty  < \varepsilon,
      \]
     we have that $\A(u) \ge \A(u_0)$. 
   \item[(DLM)] \textit{directional local minimum point} if the function 
      \[
         \varphi(s) = \A(u_0+sv),
      \]
      has a local minimum point at $s=0$ for all $v \in V$. Note that, fixed
      $v \in V$, $\varphi:(-\delta,\delta) \to \R$ is a function of the real variable
      $s$, defined for $\delta>0$ small enough.
\end{itemize}
From the classic theory of Calculus of Variations it is well-known that, if $u_0$ is a $C^2$ local minimum point, then it solves the
Euler-Lagrange equation associated to \eqref{eq:funct}, i.e.
\begin{equation}
   \frac{d}{dt} L_{\dot{u}}\big(t, u_0(t),\dot{u}_0(t)\big) = L_u\big(t,
   u_0(t),\dot{u}_0(t)\big).
   \label{eq:EL}
\end{equation}
Moreover, $u_0$ satisfies a periodic condition on the derivative, i.e. 
\begin{equation}
   L_{\dot{u}}(0,u_0(0),\dot{u}_0(0)) = L_{\dot{u}}(T,u_0(T),\dot{u}_0(T)),
   \label{eq:PBCv}
\end{equation}
which leads to
\begin{equation}
   \dot{u}_0(0) = \dot{u}_0(T),
   \label{eq:periodicVelocity}
\end{equation}
under the assumption that $L$ is globally convex in $\dot{u}$.
Note that a solution $u_0$ of \eqref{eq:EL} is a (DLM) if and only if the \textit{second variation}
\begin{equation}
   \delta^2 \A(v) = \int_{0}^{T} \bigl(v(t) \cdot \hat{L}_{uu}(t) v(t) + 2 \dot{v}(t)\cdot
   \hat{L}_{u\dot{u}}(t) v(t) + \dot{v}(t)\cdot \hat{L}_{\dot{u}\dot{u}}(t)\dot{v}(t)
   \bigr) dt,
   \label{eq:sv}
\end{equation}
is non-negative for all $v \in V$, where
\begin{gather*}
   \hat{L}_{uu}(t) = L_{uu}(t,u_0(t), \dot{u}_0(t)), \\ 
   \hat{L}_{u\dot{u}}(t) = L_{u\dot{u}}(t,u_0(t), \dot{u}_0(t)), \\ 
   \hat{L}_{\dot{u}\dot{u}}(t)= L_{\dot{u}\dot{u}}(t,u_0(t), \dot{u}_0(t)),
\end{gather*}
are the second derivatives of the Lagrangian along $u_0$.
Note that $\delta^2 \A$ is a quadratic functional, defined on the whole space of $T$-periodic
functions $V$. In the following, we recall some sufficient conditions for a solution of the
Euler-Lagrange equation to be either a (DLM), (WLM), or (SLM).

We stress out that what presented in Section~\ref{s:quadFun} and \ref{s:slm} can be
obtained as particular case of results obtained for general boundary conditions (e.g.
\cite{dosla-dosly_1997, stefanini-zezza_1997, zeidan_1999, zeidan_2001}).
Nevertheless, it is useful to recall simpler proofs specialized to the case of periodic
boundary conditions, and see how to adapt them when an additional symmetry is present.
\section{Quadratic functionals}
\label{s:quadFun}
We consider a quadratic functional
\begin{equation}
   \Q(v) = \int_{0}^{T} \big( v(t)\cdot P(t) v(t) + 2 \dot{v}(t)\cdot Q(t) v(t) +
   \dot{v}(t)\cdot R(t)
   \dot{v}(t) \big) dt,
   \label{eq:quadfunc}
\end{equation}
defined on the whole space $V$, where $P,Q,R:[0,T] \to \R^{n\times n}$ are $C^1$
matrix functions such that\footnote{The
character $T$ is already used to denote the period. However, when we use the
superscript $T$ for a matrix, we mean the transpose of the matrix itself. This notation
will not be confusing in the following, since it is always clear when we intend to
transpose a matrix.} $P(t)=P^T(t), \, R(t)=R^T(t)$ for all $t \in [0,T]$.
The Euler-Lagrange equation associated to \eqref{eq:quadfunc} is 
\begin{equation}
   \frac{d}{dt}\big(R\dot{y} + Qy\big) = Q^T \dot{y} + Py,
   \label{eq:jde}
\end{equation}
and it is usually called \textit{Jacobi differential equation}.
If $\det R(t) \neq 0$ for all $t \in [0,T]$, setting $z = R\dot{y}+Qy$, we can write the
system \eqref{eq:jde} as
\begin{equation}
   \begin{cases}
      \dot{y} = Ay+Bz, \\
      \dot{z} = Cy-A^Tz,
   \end{cases}
   \label{eq:jde1}
\end{equation}
where 
\[
   A = -R^{-1}Q, \quad B=R^{-1}, \quad C = P-Q^TR^{-1}Q.
\]
Note that $B$ and $C$ are symmetric matrices.
It is also useful to introduce the matrix version of equation \eqref{eq:jde1}, i.e.
\begin{equation}
   \begin{cases}
      \dot{Y} = AY+BZ, \\
      \dot{Z} = CY-A^TZ,
   \end{cases}
   \label{eq:jde2}
\end{equation}
where $Y, Z: [0,T] \to \R^{n \times n}$ are matrix functions.
\begin{remark}
   Note that, if $(Y_1,Z_1), (Y_2,Z_2)$ are two solutions of \eqref{eq:jde2}, then 
   \begin{equation}
      Y_1^T(t)Z_2(t) - Z_1^T(t)Y_2(t) \equiv K,
      \label{eq:constComb}
   \end{equation}
   where $K \in \R^{n \times n}$ is a constant matrix. 
   \label{rmk:constComb}
\end{remark}
\begin{definition}
A solution $(Y,Z)$ of \eqref{eq:jde2} is said to be \textit{self-conjoined} if
\begin{equation}
   Y^TZ-Z^TY \equiv 0.
   \label{eq:scs}
\end{equation}
\end{definition}
%
\begin{definition}
   Let $(y,z)$ be a non-zero solution of system \eqref{eq:jde1} such that $y(0)=0$. A point $c \in
   (0,T]$ is said to be \textit{conjugate} with $0$ if $y(c)=0$.
\end{definition}

\begin{remark}
Note that $c \in (0,T]$ is conjugate with $0$ if and only if 
$\det Y_0(c) =0$,
where $(Y_0,Z_0)$ is the solution of \eqref{eq:jde2} with initial conditions
\[
\begin{cases}
   Y_0(0) = 0, \\
   Z_0(0) = \Id.
\end{cases}
\]
\label{rmk:detNonZero}
\end{remark}

\begin{definition}
   The \textit{Legendre condition} (L) $\big($\textit{strengthened Legendre
      condition} (L')$\big)$ holds
      if $R(t) \geq 0$\footnote{When we write $A>0$ (resp. $A \geq 0$), where $A \in \R^{n \times n}$ is a symmetric
      matrix, we mean that $A$ is positive definite (resp. positive semi-definite).}
      $\big( R(t) > 0\big)$ for all $t \in
      [0,T]$.
\end{definition}

\begin{definition}
The \textit{regularity condition} (R) $\big($\textit{strengthened regularity
      condition} (R')$\big)$ holds if
      \[
          \int_{0}^{T} P(t) dt   \geq 0 \quad \bigg( \int_{0}^{T} P(t) dt > 0 \bigg).
      \]
\end{definition}

\begin{definition}
   The \textit{Jacobi condition} (J) $\big($\textit{strengthened Jacobi condition}
      (J')$\big)$ holds if every non-zero solution $(y,z)$ of
      \eqref{eq:jde1} with initial condition $y(0) = 0$ does not have any conjugate point
   $c \in (0,T)$ $\big(c \in (0, T]\big)$ with 0.
\end{definition}

In the classic setting of the fixed end-point problem, it is known that (L) is a necessary
condition for the positivity of a quadratic functional, and moreover, if (L') holds, then (J)
is also necessary, i.e. there are no
conjugate points (see e.g. \cite{giaquinta_2004, pars_1962, vanbrunt_2003}).
Here we can prove similar necessary conditions.
\begin{lemma}
   If $\Q(v)\geq 0$ for all $v \in V$, then conditions \textnormal{(L)} and \textnormal{(R)} hold. 
   \label{lemma:necessary1}
\end{lemma}
\begin{proof}
   The proof that (L) holds is the same as in the fixed end-point problem, since we can restrict
   ourself to local variations vanishing at the extrema of the interval. The regularity condition (R)
   follows by taking constant variations $v(t) \equiv v_0 \in \R^n $, since in this case
   \[
      \Q(v) = v_0 \cdot \int_{0}^{T}P(t)\, dt \, v_0 \geq 0,
   \]
   which is exactly condition (R).
\end{proof}

\subsection{Sufficient conditions for positivity}
The positivity of a quadratic functional can be expressed in terms of the existence
of a symmetric solution $W$ of the \textit{Riccati differential equation} (RDE)
\begin{equation}
   \dot{W}-C+WA+A^TW+WBW = 0,
   \label{eq:rde}
\end{equation}
with certain boundary conditions. 

\begin{remark}
Note that, if $(Y,Z)$ is a solution of \eqref{eq:jde2} such that $Y(t)$ is non-singular on
the whole $[0,T]$, then 
\[ 
   W(t) = Z(t)Y^{-1}(t),
\]
is a solution of \eqref{eq:rde} defined on the whole interval $[0,T]$. 
Moreover, if $(Y,Z)$ is also self-conjoined, then $W$ is symmetric.
\end{remark}

\begin{definition}
  Condition (SR) holds if there exists a symmetric
      solution $W(t)$ of (RDE) \eqref{eq:rde} defined on the whole interval
      $[0,T]$ and such that
      \begin{equation}
         W(T)-W(0) > 0.
         \label{eq:riccatiBC}
      \end{equation}
\end{definition}
This condition is sufficient to have a positive definite quadratic functional in
the case of periodic boundary conditions. For other types of boundary conditions, 
the inequality \eqref{eq:riccatiBC} has to be adapted (see, e.g. \cite{zeidan_2001}).
\begin{theorem}
   Let conditions $\textnormal{(L')}$ and $\textnormal{(SR)}$ hold. Then we have that
   $\Q(v)> 0$ for all non-zero
   $v \in V$. 
   \label{th:DLMsuf}
\end{theorem}
\begin{proof}
   Let $W$ be the symmetric solution of \eqref{eq:rde} defined on the
   whole interval $[0,T]$, such that $W(T)-W(0)>0$. 
   Let $v \in V$ be a non-zero $T$-periodic function, then we have 
   \[
      \begin{split}
         \Q(v) & = \int_{0}^{T}\bigg( v \cdot P v + 2 \dot{v} \cdot Q v + \dot{v} \cdot R
         \dot{v} - \frac{d}{dt}(v \cdot W v) + \frac{d}{dt}(v \cdot W v) \bigg) dt \\
         & = v(t)\cdot W(t) v(t)\bigg|_0^T + \int_{0}^{T}\big(
         R\dot{v}+Qv-Wv\big)\cdot B\big( R\dot{v}+Qv-Wv\big)dt \\
         & = v(0)\cdot\big(W(T)-W(0)\big)v(0) + \int_{0}^{T}\big(
         R\dot{v}+Qv-Wv\big)\cdot B\big( R\dot{v}+Qv-Wv\big)dt ,
      \end{split}
   \]
   where we have used that $v(0)=v(T)$ in the last equality. Since $B(t) = R^{-1}(t)$, from condition
   (L'), $B$ is positive definite for all $t \in [0,T]$, then the function in the integral
   is positive. Since also $W(T)-W(0)$ is positive definite, we have that $\Q(v) > 0$. 
\end{proof}

The following lemma relates the dimension and the sign of the determinant of $Y_0(t)$ with
the (SR) condition. This could be useful to search for a symmetric solution $W$ of the
Riccati differential equation satisfying the boundary condition $W(T)-W(0)>0$.
\begin{lemma}
   Let conditions $\textnormal{(L')}$ and $\textnormal{(J')}$ hold. Let $(Y_0, Z_0)$ be the solution of \eqref{eq:jde2}
   with initial conditions
   \[
      \begin{cases}
         Y_0(0) = 0, \\ Z_0(0) = \Id.
      \end{cases}
   \]
   Then
   \begin{itemize}
      \item[(i)] if $n$ is even and $\det Y_0(t)>0$ for $t \in (0,T]$, then condition
         $\textnormal{(SR)}$ holds; 
      \item[(ii)] if $n$ is odd and $\det Y_0(t)<0$ for $t \in (0,T]$, then condition
         $\textnormal{(SR)}$ holds.
   \end{itemize}
   \label{lemma:riccatiBC1}
\end{lemma}
\begin{proof}
   From condition (J'), we have that the solution $(Y_0,Z_0)$ of \eqref{eq:jde2}
   is such that $\det Y_0(t) \neq 0$ for all $t \in (0,T]$, then we can define $W_0(t) =
   Z_0(t)Y_0^{-1}(t)$. Let $\varepsilon>0$ and $(Y_\varepsilon,Z_\varepsilon)$ be the solution of \eqref{eq:jde2} with initial
   conditions
   \[
   \begin{cases}
      Y_\varepsilon(0) = -\varepsilon\Id, \\ Z_\varepsilon(0) = \Id.
   \end{cases}
   \]
   From the continuous dependence of the solutions with respect to the initial conditions,
   we know that $(Y_\varepsilon, Z_\varepsilon) \to (Y_0,Z_0)$ uniformly in $[0,T]$ as
   $\varepsilon\to0$.
   Moreover, 
   \[
      \det Y_\varepsilon(0) = (-\varepsilon)^n,
   \]
   hence in the hypotheses (i) or (ii) we have that $\det Y_0(t)$ and
   $\det Y_\varepsilon(t)$ have the same sign for $t$ near zero.
   Therefore, it follows that $\det Y_\varepsilon(t) \neq 0$ for
   all $t \in [0,T]$, for all $\varepsilon < \bar{\varepsilon}$ with $\bar{\varepsilon}$ small
   enough. Moreover, evaluating \eqref{eq:constComb} for $t=0$, we obtain that $(Y_\varepsilon, Z_\varepsilon)$ is self-conjoined, hence 
   $W_\varepsilon(t) = Z_\varepsilon(t)Y^{-1}_\varepsilon(t)$
   is a symmetric solution of \eqref{eq:rde} defined on the whole $[0,T]$.    

   Now we prove that there exists $\varepsilon>0$ small enough such that 
   \begin{equation}
      w\cdot\bigl(W_\varepsilon(T)-W_\varepsilon(0)\bigr)w>0,
      \label{eq:posDefThesis}
   \end{equation}
   for all $w \in \R^n \setminus \{ 0 \}$. Without loss of generality, we can prove
   \eqref{eq:posDefThesis} for all $w \in \S^{n-1} \subset \R^n$.
   First we note that
   \[
      W_\varepsilon(0) = Z_\varepsilon(0)Y_\varepsilon^{-1}(0) = -\frac{1}{\varepsilon}
      \Id. 
   \] 
   Let $w \in \S^{n-1}$ be a vector on the unit sphere, then 
   \begin{align*}
      \lim_{\varepsilon \to 0^+} w\cdot W_\varepsilon(T) w & = w \cdot W_0(T) w \in
         \R, \\
      \lim_{\varepsilon \to 0^+} w \cdot W_\varepsilon(0) w & = 
      \lim_{\varepsilon \to 0^+} \frac{1}{\varepsilon} = -\infty.
   \end{align*}
   Therefore, since the unit sphere is compact, inequality \eqref{eq:posDefThesis} is
   verified uniformly for $\varepsilon$ small enough, hence the thesis.
\end{proof}

\section{Weak and strong local minimizers}
\label{s:slm}
To discuss weak and strong local minimizers, we need few other definitions and conditions,
which are also used in the classic fixed end-point problem. 
We define the \textit{Weierstrass excess function} $E$ as
\begin{equation}
   E(t,u,v,w) := L(t,u,w) - L(t,u,v) - (w-v) \cdot L_{\dot{u}}(t,u,v).
   \label{eq:wef}
\end{equation}
Let $u_0 \in X$ be a solution of the Euler-Lagrange equation \eqref{eq:EL}. To simplify
the notations, we define the \textit{tube} around $u_0$ of radius
$\varepsilon>0$ as
\[
   T(u_0, \varepsilon) = \big\{ (t, y) \in [0,T] \times \R^n : |y-u_0(t)| < \varepsilon
   \big\}, 
\]
and the \textit{restricted tube} as
\[
   RT(u_0, \varepsilon) = \big\{ (t, y, v) \in [0,T] \times \R^n \times \R^n : |y-u_0(t)| <
      \varepsilon, \, |v-\dot{u}_0(t)|<\varepsilon \big\}.
\]
We introduce also additional conditions.
\begin{definition}
 The \textit{Weierstrass condition} (W) $\big($\textit{strengthened Weierstrass condition}
 (W')$\big)$ holds if
      \begin{equation}
         E(t,u_0(t),\dot{u}_0(t), w) \geq 0,
         \quad
         (E(t,y,v,w) \geq 0),
         \label{eq:wcondition}
      \end{equation}
      for all $t \in [0,T]$ $\big($for all $(t,y,v) \in RT(u_0,\varepsilon)$$\big)$ and for all $w \in \R^n$. 
\end{definition}

\begin{definition}
   A $C^1$ function $V(t,y)$ satisfies the \textit{Hamilton-Jacobi
      inequality} (HJ) for $v \in \R^n$ if
      \begin{equation}
         \begin{split}
         V_t(t,y)  + &V_y(t,y) \cdot v - L(t,y,v) \leq \\ &V_t(t,u_0(t)) + V_y(t, u_0(t)) \cdot
         \dot{u}_0(t) - L(t,u_0(t),\dot{u}_0(t)).
      \end{split}
         \label{eq:HJineq}
      \end{equation}

\end{definition}

\begin{remark}
   Note that condition (W') is satisfied whenever $L$ is globally convex in
   $\dot{u}$, i.e. when $L_{\dot{u}\dot{u}}\geq 0$. 
   In problems coming from classical mechanics this condition is usually fulfilled, since the velocity
   $\dot{u}$ is contained only in the kinetic energy, which is a positive definite
   quadratic form.
   \label{rmk:sw}
\end{remark} 

Sufficient conditions for both the weak and the strong local minimality can be formulated by
using the (SR) condition and the strengthened Weierstrass
condition (W'), adapting the proof for the classical end-point problem given in
\cite{clarke-zeidan_1986}. 
\begin{theorem}
   Let $u_0 \in X$ be a periodic solution of the Euler-Lagrange equation \eqref{eq:EL}. Suppose
   that conditions $\textnormal{(L')}$ and $\textnormal{(SR)}$ are satisfied for the second variation associated
   to $u_0$. Then $u_0$ is a (WLM).
   If condition $\textnormal{(W')}$ also holds, then $u_0$ is a $\textnormal{(SLM)}$.
   \label{th:SLMsuf}
\end{theorem}
\begin{proof}
   By the (SR) condition, there exists a symmetric solution $W(t)$ of the Riccati
   differential
   equation \eqref{eq:rde}, defined on the whole $[0,T]$ and such that $W(T)-W(0)>0$. 
   From the embedding theorem of differential equations (see for instance Theorem 4.1 in
   \cite{hestenes_1966}), there
   exists $\varepsilon_0>0$ and a symmetric matrix function $\widetilde{W}:[0,T] \to \R^{n
   \times n}$ such that
   \[
      \dot{\widetilde{W}}-C+\widetilde{W}A+A^T\widetilde{W}+\widetilde{W}B\widetilde{W} = -\varepsilon_{0}
      \Id,
   \]
   and $\widetilde{W}(T)-\widetilde{W}(0) > \varepsilon_0 \Id$.
   We set
   \begin{gather*}
      p(t)    = L_{\dot{u}}(t, u_0(t), \dot{u}_0(t)), \\
      V(t,y)  = p(t) \cdot y + \frac{1}{2}(y-u_0(t)) \cdot \widetilde{W}(t)(y-u_0(t)).
   \end{gather*}
   Assume for the moment that $V(t,y)$ satisfies condition (HJ) for all $(t,y ) \in
   T(u_0, \varepsilon)$ and for all $v \in \R^n$, and let $u \in X$ be another
   $T$-periodic competitor such that $\norm{u-u_0}_\infty < \varepsilon$.
   Hence, substituting $(t,u(t),\dot{u}(t))$ for $(t,y,v)$ in \eqref{eq:HJineq} and
   integrating on $[0,T]$, we get
   \[
      \begin{split}
      \int_{0}^{T}L(t,u(t),\dot{u}(t))\, dt + \big( V(0,u(0)) - V(T,u(T)) \big) \geq \\
      \int_{0}^{T}L(t,u_0(t),\dot{u}_0(t))\, dt + \big( V(0,u_0(0)) - V(T,u_0(T)) \big).
      \end{split}
   \]
   Note that, since $u(t), \, u_0(t)$ and $L(t, \cdot,\cdot)$ are $T$-periodic functions,
   then also $p(t)$ is $T$-periodic. Therefore, we have that
   \[
   \begin{split}
      V(0,u_0(0)) - V(T,u_0(T)) & = p(0) \cdot u_0(0) + \frac{1}{2} (u_0(0) - u_0(0))
      \cdot \widetilde{W}(0) (u_0(0)-u_0(0)) \\
      & - p(T) \cdot u_0(T) - \frac{1}{2} (u_0(T) - u_0(0))
      \cdot \widetilde{W}(T) (u_0(T)-u_0(0)) \\
      & = 0, 
   \end{split}
\]
\[
   \begin{split}
      V(0,u(0)) - V(T,u(T)) & = p(0) \cdot u(0) + \frac{1}{2} (u(0) - u_0(0))
      \cdot \widetilde{W}(0) (u(0)-u_0(0)) \\
      & - p(T) \cdot u(T) - \frac{1}{2} (u(T) - u_0(0))
      \cdot \widetilde{W}(T) (u(T)-u_0(0)) \\
      & = \frac{1}{2}\bigg( (u(0)-u_0(0))\cdot \widetilde{W}(0)(u(0)-u_0(0)) \\& -
      (u(0)-u_0(0))\cdot \widetilde{W}(T)(u(0)-u_0(0))  \bigg) \\
      & < 0.
   \end{split}
  \]
  Hence, the inequality above implies 
  \[
      \int_{0}^{T}L(t,u_0(t),\dot{u}_0(t))\, dt  \leq  \int_{0}^{T}L(t,u(t),\dot{u}(t))\,
      dt, 
  \]
  i.e. $u_0$ is a (SLM). 

  The proof that $V(t,y)$ satisfies (HJ) is the same as the one in
  \cite{clarke-zeidan_1986}. If condition (W') is dropped, condition (HJ)
  is satisfied only on a restricted tube $RT(u_0,\varepsilon)$, for
  some $\varepsilon>0$, and therefore $u_0$ is only a (WLM). 

\end{proof}
\section{Application to Celestial Mechanics problems}
\label{s:locminEx}
In this last section we show some examples of application of the above results to
problems of Celestial Mechanics.
We first consider circular solutions of the Kepler problem with potentials of
type $1/r^\alpha$, where $r$ is the distance from the origin and $\alpha > 0$.
%
Second, we take into account the figure-eight solution of the 3-body problem
\cite{chenciner-montgomery_2000}.


\subsection{Kepler problem with $\alpha$-homogeneous potential}
   Given $T>0$ and $\alpha>0$, we consider the action of the Kepler problem
   with $\alpha$-homogeneous potential
   \begin{equation}
      \A^\alpha(u) = \int_{0}^{T}\bigg( \frac{1}{2} \lvert \dot{u} \rvert^2 +
      \frac{1}{\lvert u \rvert^\alpha} \bigg) dt,
      \label{eq:kepProb}
   \end{equation}
   defined on the set 
   \begin{equation}
      X_k = \{ u \in H^1_T([0,T], \R^2 \setminus \{0\}): \text{deg}(u,0)=k \},
      \label{eq:Xk}
   \end{equation}
   where $H^1_T([0,T], \R^2 \setminus \{0\})$ is the space of $T$-periodic $H^1$ functions
   that do not intersect the origin, and $k \in \Z$ is an integer. 
   The equation of motion associated to the functional \eqref{eq:kepProb} is 
   \begin{equation}
      \ddot{u} = -\alpha\frac{u}{\lvert u \rvert^{(2+\alpha)}},
      \label{eq:2bpAlpha}
   \end{equation}
   and the coefficients of the second variation $\delta^2 \A^\alpha$ for a solution of
   \eqref{eq:2bpAlpha} are 
   \begin{equation}
      R(t) = L_{\dot{u}\dot{u}} = \Id, \quad Q(t) = L_{u\dot{u}} = 0, \quad 
      P(t) = L_{uu} = -\alpha \frac{\Id}{|u|^{(2+\alpha)}} + \alpha(\alpha+1)\frac{u
      u^T}{|u|^{(4+\alpha)}}.
      \label{eq:kep2var}
   \end{equation}
   For each $\alpha > 0$, there exists a $T$-periodic circular orbit given by
   \[
      u_0(t) = \big(a \cos(nt), a \sin(nt)\big), \quad n=\frac{2\pi}{T}, \quad a =
      \bigg(\frac{\alpha}{n^2}\bigg)^{\frac{1}{2+\alpha}}.
   \]
   In the Keplerian case $\alpha=1$, it is known that \eqref{eq:kepProb}
   attains its global minimum at the elliptical $T$-periodic functions satisfying
   the Keplerian equations of motion \eqref{eq:2bpAlpha}, and for which $T$ is the minimum
   period (see \cite{gordon_1977}).
   In \cite{venturelli_thesis}, the author generalized the result of \cite{gordon_1977},
   proving that
   \begin{itemize}
   \item[(i)] if $k = \pm 1$ and $\alpha \in (1,2)$, then the minimizers of $\A^\alpha$ on
         $X_k$ are the circular orbits;
   \item[(ii)] if $k \ne 0$ and $\alpha \in (0,1)$, then the minimizers of $\A^\alpha$ on
         $X_k$ are the collision-ejection solutions. 
   \end{itemize}
   In the proof of (ii) however, it is not mentioned the type of stationary point of the
   circular orbit. 
   Moreover, the author stated that finding the global minimizer of $\A^\alpha$ on
   $X_k$ for $|k|\geq 2$ and $\alpha \in (1,2)$ is still an opened problem.
   Therefore, the Kepler problem with $\alpha$-homogeneous potentials is a good benchmark to produce non-trivial examples for
   studying minimality properties of periodic solutions.

   \paragraph{Computations for $k=1$.} 
   Since the coefficients of the Jacobi differential equation \eqref{eq:jde2} 
   depend directly on the time, we used a numerical integrator to compute the 
   solution $(Y_0, Z_0)$ corresponding
   to the circular orbit with period $T= 2\pi$. The computations were performed for 
   $\alpha = 0.2, 0.4, 0.6, 0.8, 1, 1.2, 1.4, 1.6$, and the plot of the determinant of
   $Y_0$ as a function of time is shown in Figure~\ref{fig:detY0}.  
   \begin{figure}[!ht]
      \centering
         \includegraphics[scale=0.5]{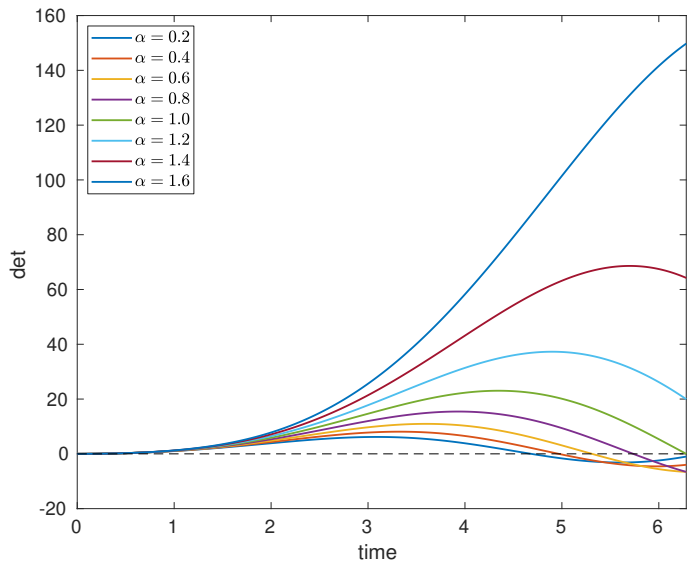} 
      \caption{The determinant of the matrix $Y_0(t)$ for the
      circular orbits with period $T=2\pi$, for different values of $\alpha$.}
      \label{fig:detY0}
   \end{figure}

   For the case $\alpha =1$, a conjugate point appears exactly at the end point $T=2\pi$ 
   of the time interval, hence the Jacobi condition (J) holds, but not the strengthened Jacobi 
   condition (J'). Therefore, the gravitational Kepler problem provides an example 
   where we can find global minimizers for which the sufficient conditions stated above
   for the weak local minimality are not satisfied, and the second variation is only non-negative definite.
   This is consistent with the result provided in \cite{gordon_1977}, because the circular
   orbit is embedded in a family of periodic solutions of the Kepler problem with the same
   period. This degeneration is reflected in the fact that the second variation $\delta^2
   \A^1$ is only non-negative.

   For the case $\alpha > 1$ there are no conjugate points in $(0,T]$, hence the strengthened
   Jacobi condition (J') holds for the second variation.
   The determinant of $Y_0(t)$ is greater than zero in $(0,T]$, and the dimension of the system
   is $n=2$, hence we are in the hypotheses of Lemma~\ref{lemma:riccatiBC1} and condition
   (SR) is therefore satisfied. It follows that the second variation of the circular orbit
   is positive definite, hence it is a (DLM). Moreover, the Lagrangian of the functional
   \eqref{eq:2bpAlpha} comes from a mechanical system and it is globally convex in the
   velocity $\dot{u}$, hence by Remark~\ref{rmk:sw} we know that
   the strengthened Weierstrass condition (W') holds. Therefore, by
   Theorem~\ref{th:SLMsuf} we are able to conclude that the circular orbit is a (SLM) on
   $X_1$. Note that this was expected, because we know that circular orbits are the global
   minimizers.

   For $\alpha \in (0,1)$ a conjugate point appears inside the interval $(0,T)$ in the examples with
   $\alpha=0.2, \, 0.4, \, 0.6, \,0.8$. Hence circular solutions are not even local
   minimizers, but rather saddle points. This provides more information than what was proved in
   \cite{venturelli_thesis}, where the author showed that the action of the circular orbit
   is greater than the action of the collision-ejection solution.

   \paragraph{Computations for $k=2$.} 
   As said above, finding the global minimizer of $\A^\alpha$ on $X_k$ is still an opened
   problem for $|k| \geq 2$. We provide here some computations for $k=2$ and $\alpha =
   1.2, 1.3, 1.4, 1.5, 1.6, 1.7, 1.8$. Figure~\ref{fig:detY0_k2} shows the determinant of
   $Y_0(t)$ relative to the circular orbit with period $2\pi$, in the time interval
   $[0, 4\pi]$. 
   \begin{figure}[!ht]
      \centering
         \includegraphics[scale=0.5]{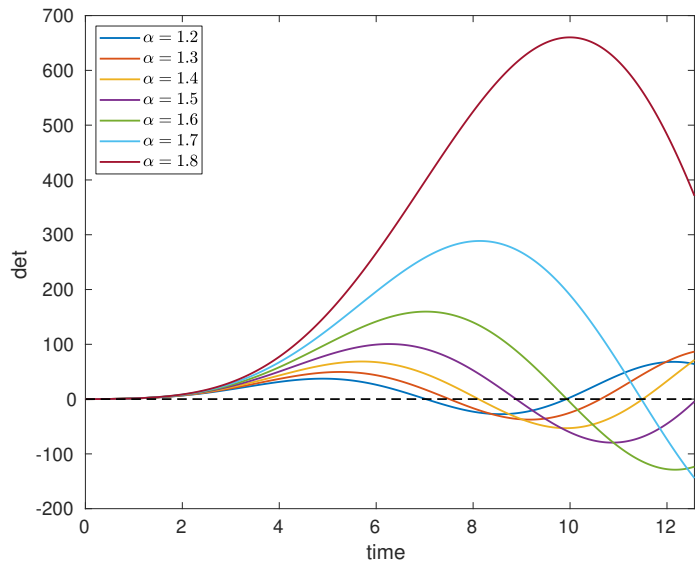} 
      \caption{The determinant of the matrix $Y_0(t)$ for the
      circular orbits with period $2\pi$ in the time interval $[0,4\pi]$, for different
   values of $\alpha>1$.}
      \label{fig:detY0_k2}
   \end{figure}
   We can notice that the circular solution has at least a conjugate point in $(0, 4\pi)$
   for $\alpha \leq 1.7$, hence the Jacobi condition (J) does not hold. Therefore it is
   not a minimizer anymore, but rather a saddle point. 

   To understand if a situation different from the case $k=1$ can occur, we can
   compare the action on $[0, 4\pi]$ (that we denote with $\A^\alpha_k, k=2$) of the circular
   solution $u_0$ with period $2\pi$, with the action of the collision-ejection solution
   $\bar{u}$. The action of the collision-ejection solution is
   \begin{equation}
      \A^\alpha_k(\bar{u}) = \bigg(
      \frac{2+\alpha}{2-\alpha}\bigg)(2\alpha^2)^{-\frac{\alpha}{2+\alpha}}T^{\frac{2+\alpha}{2-\alpha}}\bigg(
      \int_0^{2\pi}|\sin t|^{2/\alpha}dt
      \bigg)^{\frac{2\alpha}{2+\alpha}},
      \label{eq:ceKep}
   \end{equation}
   while the action of a $k$-circular solution is
   \begin{equation}
       \A^\alpha_k(u_0) = 
       k^{\frac{2\alpha}{2+\alpha}} (2+\alpha)\bigg(\frac{T}{2}\bigg)^{\frac{2+\alpha}{2-\alpha}}\bigg(
       \frac{\pi^2}{\alpha}
      \bigg)^{\frac{\alpha}{2+\alpha}},
      \label{eq:kcirc}
   \end{equation}
   see \cite{venturelli_thesis} for the details. 
   Table~\ref{tab:actionValues} reports the values obtained from numerical evaluation
   of \eqref{eq:ceKep} and \eqref{eq:kcirc}, for $T=4\pi$ and $k=2$.   
   \begin{table}[!ht]
      \centering
      \begin{tabular}{|c|ccccccc|}
         \hline
         $\alpha$                & 1.2 & 1.3 & 1.4 & 1.5 & 1.6 & 1.7 & 1.8 \\
         \hline
         $\A^\alpha_2(u_0)$      & 18.777& 18.698 & 18.599 & 18.483 & 18.355 & 18.218 & 18.073\\
         \hline
         $\A^\alpha_2(\bar{u})$  & 12.585& 13.286 & 14.349 & 15.974 & 18.563 & 23.057 & 32.281\\
         \hline
      \end{tabular}
      \caption{The value of the action $\A^\alpha_2$ of the circular solution $u_0$ with
         period $2\pi$ and of the collision-ejection solution $\bar{u}$.}
      \label{tab:actionValues}
   \end{table}
   Interestingly, we have that
   \begin{equation}
      \A^\alpha_2(u_0) < \A^\alpha_2(\bar{u})
      \label{eq:ineqcecirc}
   \end{equation}
   for $\alpha = 1.6, 1.7, 1.8$, meaning that the global minimizer is not a collision
   solution. On the other hand, for $\alpha=1.6, 1.7$ the 2-circular solution $u_0$ is not a
   local minimizer because condition (J) does not hold, hence necessarily the global
   minimum is achieved by a non-collision and non-circular $T$-periodic solution.
   These simple examples already show that finding the global minimum of $\A^\alpha$ of $X_k$ 
   for $\alpha \in (1,2)$ and $|k| \geq 2$ might be more complicated than the case of
   $k=\pm1$. 

   \subsection{The figure-eight solution of the 3-body problem}
   The figure-eight solution of the 3-body problem has been found first in
   \cite{moore_1993} by using numerical methods. Later on, in
   \cite{chenciner-montgomery_2000} the authors were able to give a proof of the existence
   of such orbit by minimizing the action \eqref{eq:NbodyAction} over a particular set of
   loops. More accurate numerical studies were performed in \cite{simo_2000, simo_2001,
   simo_2002}, and the linear stability was finally proved by using rigorous numerical
   techniques in \cite{kapela-simo_2007}. 
   From the numerical point of view, the figure-eight solution is computed in two steps
   (see e.g. \cite{fenucci-gronchi_2018, fenucci-jorba_2020} for details):
   \begin{enumerate}
      \item the action \eqref{eq:NbodyAction} is discretized by using truncated Fourier series,
         and a gradient descent method is applied to an eight-shaped first guess curve;
      \item a shooting method is applied to the output of the gradient descent method, and
         an accurate initial condition is computed.
   \end{enumerate}
   However, this procedure does not ensure that the final solution is actually a minimizer of the
   action. 
   The initial conditions computed with this method are (see
   \cite{chenciner-montgomery_2000})
   \[
      u_1 = (0.97000435669734, -0.24308753153583), \quad u_2 = -u_1, \quad u_3 = (0, 0),
   \]
   \[
      \dot{u}_3 = (-0.93240737144104, -0.86473146092102), \quad \dot{u}_1 =
      -\frac{\dot{u}_3 }{2}, \quad \dot{u}_2 = -\frac{\dot{u}_3 }{2}, 
   \]
   while the period corresponds to $T \simeq 6.32591398292621$. The figure-eight solution
   also has a dihedral symmetry, meaning that it satisfies
   \begin{equation}
      \begin{cases}
         u_1(t+T/6) = G u_2(t),\\
         u_2(t+T/6) = G u_3(t),\\
         u_3(t+T/6) = G u_1(t),
      \end{cases}
      \qquad 
      G = 
      \begin{pmatrix}
         -1 & 0  \\
         0  & 1 \\
      \end{pmatrix},
      \label{eq:extrasym8}
   \end{equation}
   and 
   \begin{equation}
      \begin{cases}
         u_1(t) = u_1(-t),\\
         u_2(t) = u_3(-t),\\
         u_3(t) = u_2(-t).
      \end{cases}
      \label{eq:timerevsym}
   \end{equation}
   Figure~\ref{fig:fig8} shows the trajectory and the initial configurations of the
   masses, obtained by integrating the 3-body problem using the above initial conditions.
   \begin{figure}[!ht]
      \centering
      \includegraphics[width=0.5\textwidth]{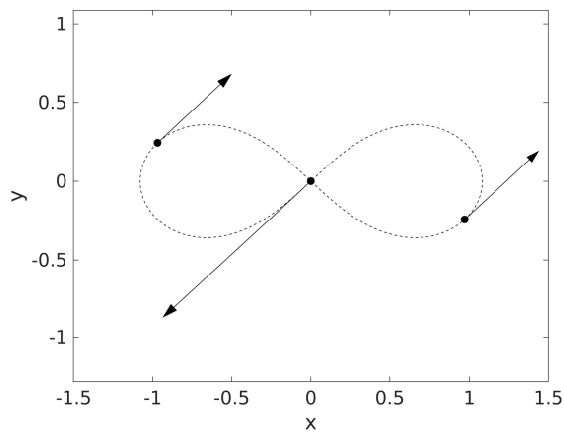}
      \caption{The trajectory and the initial configuration of the masses for the
      figure-eight solution of the 3-body problem.}
      \label{fig:fig8}
   \end{figure}

   By using a numerical integrator, we computed the solution $(Y_0, Z_0)$ of the Jacobi differential equation 
   on the whole timespan $[0,T]$. The determinant of $Y_0$ is plotted as a function of the
   time in Figure~\ref{fig:det8}.
   \begin{figure}[!ht]
      \centering
      \includegraphics[width=0.5\textwidth]{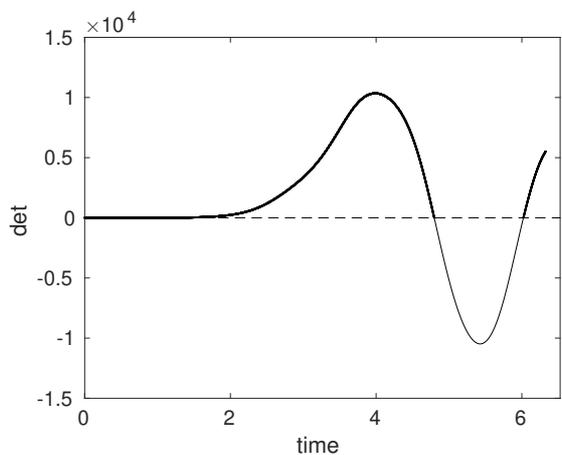}
      \caption{The determinant of $Y_0$ for the solution of the Jacobi differential
      equation corresponding to the figure-eight solution. The black thick curve
   corresponds to the intervals where $\det Y_0(t)>0$, while the thin black line
corresponds to the interval where $\det Y_0(t)<0$.}
      \label{fig:det8}
   \end{figure}
   A conjugate point appears, and therefore condition (J') does not hold.
   This means that the figure-eight is not a local minimizer of the action on the whole
   set of $T$-periodic loops (note that this result was already found in
   \cite{chtcherbakova_2006}), but rather a saddle point.

   It is worth stressing out that the proof of the existence of the figure-eight solution
   is made by proving that is minimizes the action \eqref{eq:NbodyAction}. However, the
   set on which it is a minimizer is not the whole space of $T$-periodic planar loops, 
   but rather on the subset $X$ of $T$-periodic loops fulfilling the symmetries 
   \eqref{eq:extrasym8} and \eqref{eq:timerevsym}. 
   The theory presented in Section~\ref{s:quadFun} and \ref{s:slm} is done by using the
   complete space of the $T$-periodic loops. This means that we also allow non-symmetric
   variations, that break the symmetry condition. For this reason, it is 
   possible that we can choose a non-symmetric variation that reduces the value of the
   action, and this is not in contradiction with the figure-eight being a minimizer on 
   the loop set $X$, that includes the symmetries. What we can do is to adapt the results 
   of Section~\ref{s:quadFun} and \ref{s:slm} by taking into the additional symmetry. 

   \subsubsection{Including the symmetry in the theory}
   Here we include the symmetry in the space of loops. We are not going to present again the 
   theory of local minimizers, but we underline the major changes to make in the
   conditions and the proofs. From now on, we always assume that condition (L') holds. 
   Let us suppose that the loops satisfy the condition 
   \begin{equation}
      u\bigg(t +\frac{T}{M}\bigg) = R u(t), \quad t \in [0,T],
      \label{eq:extrasym}
   \end{equation}
   where $R \in O(n)$ is a fixed orthogonal matrix and $M \in \N$ is an integer number. To make the discussion simpler, we suppose that $L$ does
   not depend on the time and
   \begin{equation}
      \int_{0}^{T} L(u,\dot{u})\,dt=M\int_{0}^{\frac{T}{M}}L(u,\dot{u}) \, dt.
      \label{eq:symmLagrangian}
   \end{equation}
   We therefore consider the functional
   \[
      \mathcal{F}(\bar{u}) = \int_{0}^{T/M}L(\bar{u},\dot{\bar{u}}) \, dt,
   \]
   defined on the set of loops $\bar{u}:[0,\, T/M ]\to \R^{n}$
   such that $R\bar{u}(0) = \bar{u}(T/M)$.
   Note that by means of \eqref{eq:symmLagrangian}, if $u_0:[0,T] \to \R^n$ is a minimizer of the functional $\A$, then
   the restriction 
   \[
      u_0\big\lvert_{[0,T/M]}:\bigg[0,\,\frac{T}{M}\bigg]\longrightarrow \R^{n},
   \]
   is a minimizer of $\mathcal{F}$. Vice versa, if $\bar{u}_0:[0,T/M]\to \R^{n}$ is a
   minimizer of $\mathcal{F}$, then we can extend it to a closed loop $u_0:[0,T] \to
   \R^{n}$ by using the symmetry \eqref{eq:extrasym}, 
   and we obtain a minimizer for $\A$. 
   Therefore, we study the minimality of $u_0$ restricted to the interval
   $[0,T/M]$ as stationary point of the functional $\mathcal{F}$. 
   
   To understand what is the equivalent condition of the (SR) condition, we follow the steps of the
   proof of Theorem~\ref{th:DLMsuf}. We still
   write the quadratic functional using a symmetric solution $W$ defined on $[0, T/M]$ and, integrating by parts,
   the term outside the integral becomes
   \[
   \begin{split}
      v(t) \cdot W(t) v(t)\bigg|_0^{T/M} & = v(T/M)\cdot W(T/M) v(T/M) - v(0) \cdot
      W(0) v(0) \\ 
      & = v(0)\cdot \bigg( R^T W(T/M) R - W(0) \bigg) v(0).
   \end{split}
   \]
   Therefore, a symmetric solution $W$ of the Riccati differential equation satisfying the boundary
   condition
   \begin{equation}
      R^T W(T/M)R - W(0)>0,
      \label{eq:SRsym}
   \end{equation}
   is sufficient to ensure the positivity of a quadratic functional. (SR) condition is
   then replaced by (SR*), i.e. there exists a symmetric solution $W$ of the (RDE)
   defined on the whole $[0,T/M]$ such that \eqref{eq:SRsym} holds.
   
   The equivalent of Lemma~\ref{lemma:riccatiBC1} is obtained by replacing $T$ with
   $T/M$ and (SR) with (SR*). The proof remains the same if we notice that the map 
   $w \mapsto Rw$ is invertible and maps the unit sphere onto itself. 
   
   Regarding the weak and the strong local minimality, the proof of Theorem \ref{th:SLMsuf}
   remains the same if we assume that 
   \begin{equation}
      p(0) - R^T p\bigg( \frac{T}{M} \bigg) = 0,
      \label{eq:mod3}
   \end{equation}
   where $p(t) = L_{\dot{u}}(t, u_0(t), \dot{u}_0(t))$.
   Note that if the derivative $L_{\dot{u}}$ is such that 
   \begin{equation}
      L_{\dot{u}}(Ru, Rv) = R L_{\dot{u}}(u,v),
      \label{eq:extraSymLag}
   \end{equation}
   for all $(u,v)$, then condition \eqref{eq:mod3} is verified, and the remaining part of the proof 
   of Theorem \ref{th:SLMsuf} is the same. 

   \paragraph{Implications for the figure-eight orbit.} 
   Denoting with $O_2$ the $2 \times 2$ matrix containing only zeros, and setting
   \begin{equation}
      R = 
      \begin{pmatrix}
         O_{2} & G & O_{2} \\ 
         O_{2} & O_{2} & G \\
         G & O_{2} & O_{2}
      \end{pmatrix} \in O(6),
      \label{eq:extrarot}
   \end{equation}
   the symmetry \eqref{eq:extrasym8} of the figure-eight can be written as
   $u(t+T/6) = R u(t)$, where $u(t) = \big( u_1(t), u_2(t), u_3(t) \big) \in \R^6$.
   From Figure~\ref{fig:det8}, we can see that the determinant of $Y_0(t)$ is positive in
   the whole interval $[0,T/6]$. Since the orbit is planar, the dimension $n=6$ of the
   system is even, and by applying the corresponding version of Lemma~\ref{lemma:riccatiBC1} we find
   that condition (SR*) holds. Therefore, by the corresponding version of Theorem~\ref{th:DLMsuf}
   the figure-eight solution is a (DLM) over the
   space $X$ of $T$-periodic loops satisfying the symmetry condition \eqref{eq:extrasym8}. 

   Moreover, condition \eqref{eq:extraSymLag} is trivially satisfied for the Lagrangian of the
   $N$-body problem, and the Weierstrass condition (W') holds because of
   Remark~\ref{rmk:sw}. Therefore, we can apply the corresponding version of
   Theorem~\ref{th:SLMsuf}, finding that the figure-eight solution is a (SLM) on the same set of symmetric
   loops $X$. 

\section*{Acknowledgments}
The author wishes to thank V. Zeidan for the suggestions about the literature, 
and G. F. Gronchi for his useful comments.
The author has been partially supported by the MSCA-ITN Stardust-R, Grant Agreement n. 813644 under the
H2020 research and innovation program.
\section*{Data availability statement}
The datasets generated during and/or analysed during the current study are available from the corresponding author on reasonable request.
\bibliography{mybib}{} 
\bibliographystyle{plain}
\end{document}